\documentclass[10pt,a4paper]{article}

\usepackage{amsmath,amssymb,amsfonts,amsthm}

\usepackage{t1enc}
\usepackage[english]{babel}
\usepackage[latin1]{inputenc}
\usepackage{color}
\usepackage{verbatim}
\usepackage{hyperref}
\usepackage{color}
\usepackage[normalem]{ulem}
\newtheorem{theorem}{Theorem}[section]
\newtheorem{lemma}[theorem]{Lemma}
\newtheorem{proposition}[theorem]{Proposition}
\newtheorem{definition}[theorem]{Definition}
\newtheorem{corollary}[theorem]{Corollary}

\theoremstyle{definition}

\usepackage{enumerate}

\renewcommand{\t}[1]{\tilde{#1}}
\renewcommand{\c}[1]{{\cal{#1}}}

\newcommand{\ra}{\rightarrow}
\newcommand{\mR}{\mathbb{R}}

\newcommand{\be}{\begin{equation}}

\newcommand{\ee}{\end{equation}}

\newcommand{\Rt}{\mathbb{R}^3}

\title{Sobolev spaces for multi-black hole initial data}

\author{Mar\'\i a   E. Gabach-Cl\'ement$^{1}$ and Andr\'es Ace\~na$^{2}$ \\
  $^1$FaMAF, UNC, IFEG, CONICET, C\'ordoba, Argentina.  \\
  $^{2}$FCEN,UNCuyo, CONICET, Mendoza, Argentina. }

\begin{document}
 \maketitle

\begin{abstract}
In this article we introduce weighted Sobolev spaces that are well suited to treat initial data for multiple black hole systems. We prove general results for elliptic operators 
on these spaces and give a simple proof of existence of a class of initial data describing many extremal black holes. 
\end{abstract}

\section{Introduction}

The question of existence of initial data for isolated multiple black hole systems  can be best analyzed according to whether any of the black holes is extremal or not. This is related
to the different topologies in the initial manifold. For non-extremal black holes, the initial manifold contains one asymptotically flat (AF) end for each black hole, but for 
 extremal black holes, the associated ends are asymptotically cilindrical (AC) \cite{dain10d}. This change in topology translates into different behaviours of the initial data 
 near the ends, and results in that 
 for extremal black holes, the familiar theorems for elliptic 
equations that use standard weighted Sobolev spaces on AF manifolds, need to be adapted.

The treatment of the Lichnerowicz equation (arising from the conformal method) in the non-extremal case, \textit{i.e.} initial manifold only
with AF ends, is  similar to the case of 
compact manifolds, see for instance \cite{Choquet99}, \cite{Chrusciel02a}, \cite{Maxwell:2004yb}, \cite{Maxwell:2005bd} (cf. \cite{Chrusciel:2010my}). 
We focus here on the extremal case, that is, when there is one AF end representing the region at 
spatial infinity, and at least two AC ends, representing the extremal black holes.

In the literature, the many ends are usually treated individually, with cutoff functions that single out one asymptotic end at a time.
In this line, Chrusciel \textit{et al}, \cite{Chrusciel:2012np}, prove that if the scalar curvature of the seed metric is 
positive, a solution to the Lichnerowicz equation on a manifold with a finite number of AF and AC ends, exists. 
This is particularly relevant after the result of Leach \cite{Leach:2016jra} on manifolds with one AF and one AC end that states that if the 
Riemannian manifold is Yamabe positive, then there exists a conformal transformation to a metric with positive scalar curvature 
(see Proposition 3.5 in \cite{Leach:2016jra}). Also in the same article, an existence result for the existence of far from Constant Mean Curvature (CMC) initial data 
describing an isolated black hole is given.

Sobolev spaces arise naturally when studying the elliptic equations coming from the Einstein constraints. Concerning manifolds with one AF and one AC end, Bartnik \cite{Bartnik86} introduced weighted Sobolev spaces $W'^{k,p}_\delta$, that not only 
describe asymptotically flat solutions, but that also single out a point in the initial 
manifold in such a way that by an appropriate choice of $\delta$, the origin in $\mathbb R^3$ may be made to describe another asymptoric end. These weighted spaces 
were crucially used in \cite{Dain:2010uh,gabach09} for studying CMC, single, extremal black hole initial data, and further results concerning 
these spaces were proven.

In this article we present weighted Sobolev spaces with weights that reflect the existence of more than one black hole. We also extend the 
results in \cite{Bartnik86, Dain:2010uh,gabach09}. The aim of these results is to provide an adequate framework  to study initial 
data for many black holes. In particular we use this framework to study existence of initial data for Einstein equations that 
describe many extremal black holes.

We will consider initial data with one AF end and several AC ends. The initial data is the set $(M,g_{ij},K_{ij}, E^i,B^i)$ subjected to the constraint equations
\be
R+K^2-K_{ij}K^{ij}=2(E_iE^i+B_iB^i),
\ee
\be\label{mom}
D_jK^j\,_i-D_iK=0,
\ee
\be\label{maxcon}
D_iE^i=0,\qquad D_iB^i=0,
\ee
where $K=K_{ij}g^{ij}$, $D_i$ and $R$ are respectively the covariant derivative and 
the curvature scalar associated to the metric $g_{ij}$.

In the Conformal Method with constant mean curvature, we consider the rescaling
\be\label{rescaling}
g_{ij}=\Phi^4\tilde g_{ij},\quad K_{ij}=\Phi^{-2}\tilde K_{ij},\quad E^i=\Phi^{-6}\tilde E^i,\quad B^i=\Phi^{-6}\tilde B^i.
\ee
The constraint equations in terms of the conformal quantities are
\begin{eqnarray}
\label{constEq}   \t{\Delta}\Phi = \frac{1}{8}\t{R}\Phi+\frac{\t{K}^2-\t{K}_{ij}\t{K}^{ij}}{8\Phi^7}-\frac{\t{E}_i\t{E}^i+\t{B}_i\t{B}^i}{4\Phi^3}, \\
  \t{D}_j\t{K}^j\,_i-\t{D}_i\t{K}+4\Phi^{-1}\t{K}\t{D}_i\Phi = 0, \\
  \t{D}_i\t{E}^i = 0, \\
  \t{D}_i\t{B}^i = 0.
\end{eqnarray}

As a concrete example, one of the most relevant multi-black hole initial data was found by Majumdar and Papapetrou \cite{Papapetrou45,Majumdar47} 
and consists of $N$ charged black holes with equal mass and charge parameters for each black hole. It is interpreted as a set of extremal Reissner-Nordstr\"om punctures, 
held in equilibrium
by the balance between the gravitational atraction and the electrostatic repulsion.
The space metric and electromagnetic potential are given by
\be
ds^2=\Phi^{4}ds^2_{flat},
\ee

\be
A=\pm \Phi^{-2}dt,
\ee
where $ds^2_{flat}$ is the 3-metric of Euclidean space, and 
\be
\Phi=\sqrt{1+\sum_i^N\frac{m_i}{|x-x_i|}},
\ee
where $m_i, \; i=1,\ldots,N$, is the positive electric charge of the $i-$th black hole located at $x=x_i$.
This solution contains $N+1$ ends, the end $|x|\to\infty$ is AF and the $N$ ends 
$x\to x_i$ are AC. The latter can be seen by taking the limit of the metric and ckecking that it goes to the standard metric on the cylinder
(see 
\cite{Acena:2015rxa} for details).

The article is organized as follows. In section \ref{Sobolev} we define the functional spaces needed for dealing with many AC ends. We prove estimates, study the behaviour of functions and properties of operators on these spaces. 
In section \ref{initial} we apply these results to study perturbations 
of axially symmetric initial data for many extremal black holes.

\section{Weighted Sobolev spaces}\label{Sobolev}

The functional spaces that we need in order to deal with many asymptotic ends are extensions of the standard weighted Lebesgue and Sobolev spaces. For comparison we state the norms of those spaces as
presented by Bartnik \cite{Bartnik86}. Given locally $p$-measurable functions in $\mathbb R^n$ and $\mathbb R^n\setminus\{0\}$ respectively, define
\be\label{LLnorm00}
 L^p_\delta:\qquad ||u||_{p,\delta} := \left\{\begin{array}{cc}
                                  \left(\int_{\mR^n}|u|^p\, \sigma^{-\delta p -n} dx\right)^\frac{1}{p},&p<\infty\\
                                  \mbox{ess sup}_{\mathbb R^n}(\sigma^{-\delta}|u|),& p=\infty           
                                           \end{array}\right.
\ee
\be\label{Lnorm00}
  W^{k,p}_\delta:\qquad ||u||_{k,p,\delta} := \sum_{j=0}^k ||D^ju||_{p,\delta-j},
\ee
 
\be\label{LLnorm0}
 L'^p_\delta:\qquad ||u||'_{p,\delta} := \left\{\begin{array}{cc}
                                  \left(\int_{\mR^n\backslash\{0\}}|u|^p\, r^{-\delta p -n} dx\right)^\frac{1}{p},&p<\infty\\
                                  \mbox{ess sup}_{\mathbb R^n\setminus\{0\}}(r^{-\delta}|u|),& p=\infty           
                                           \end{array}\right.
\ee
\be\label{Lnorm0}
  W'^{k,p}_\delta:\qquad ||u||'_{k,p,\delta} := \sum_{j=0}^k ||D^ju||'_{p,\delta-j},
\ee
where $\sigma:=(1+r^2)^{1/2}$. Both $\sigma$ and $r$ are in $L^1_{loc}(\mathbb R^3)$ and $L^1_{loc}(\mathbb R^3\setminus\{0\})$ respectively. The unprimed set \eqref{LLnorm00}, \eqref{Lnorm00} 
is best suited to problems involving AF initial data, with only one asymptotic end. For single black hole initial data, with two asymptotic ends, 
one AF and one at the origin, the spaces \eqref{LLnorm0}, \eqref{Lnorm0} are more appropriate, as one locates the black hole puncture precisely at $r=0$. 
In order to extend these spaces to account for many punctures we introduce the following definitions.

\begin{definition}\label{def1}
Let $\Pi:=\{ x_i\in \mR^n,\;i=1,\ldots,N\}$,  and the function $w\in L^1_{loc}(\mR^n\setminus \Pi)$ be
\be
 w:= \left(\sum_{i=1}^N\frac{1}{r_i}\right)^{-1},
\ee
where $r_i:=|x-x_i|$ is the Euclidean distance from the point $ x_i$.
Then we define the norms and associated functional spaces 
\be\label{Lnorm}
L'^p_{w,\delta}:\qquad ||u||'_{w,p,\delta} := \left\{\begin{array}{cc}
                                  \left(\int_{\mR^n\setminus\Pi}|u|^p\, w^{-\delta p -n} dx\right)^\frac{1}{p},&p<\infty\\
                                  \textnormal{ess sup}_{\mathbb R^n\setminus\Pi}(w^{-\delta}|u|),& p=\infty           
                                           \end{array}\right.
\ee

\be\label{Wnorm}
  W'^{k,p}_{w,\delta}:\qquad ||u||'_{w,k,p,\delta} := \sum_{j=0}^k ||D^ju||'_{w,p,\delta-j}.
\ee
\end{definition}
The particular choice of $w$ responds to its behaviour near the $N$ punctures and the extra asymptotic end.
If we approach one of the punctures, say $l$, we have $r_l \ra 0$, and then
\be
 w \ra r_l.
\ee
Also, if $r_l \ra \infty$, then all $r_i \ra\infty$, and
\be
 w \ra \frac{r}{N}.
\ee
It is clear that $w$ reduces to $r$  and the norms \eqref{Lnorm}-\eqref{Wnorm} reduce to \eqref{LLnorm0}-\eqref{Lnorm0} when there is only one puncture at the origin ($N =1,\,r_1=r$). 
Moreover, Definition \ref{def1} treats each puncture in the same manner, \textit{i.e}. there is one $\delta$ for all the punctures. This is particularly relevant for the case where
all the black holes are extremal. If some of the ends are AC and some are AF (apart from the region at infinity), then a different weight should be used, with specific behaviour near 
the two types of punctures.

The importance of the behaviour of the function $w$ lies in making the standard norms and the new norms equivalent when we restrict the domain to a neighborhood of only one end. For this, we need to recenter the standard norms with the origin in the corresponding puncture. To show that they are indeed equivalent we need to isolate each end, including the $r\ra\infty$ end, and therefore we define the following quantities
\be
d_{ij} := | x_i- x_j|,
\ee
\be
R_i := \frac{1}{3}\min_{j\neq i}\{d_{ij}\},
\ee
\be
R := 3\max_{i}\{| x_i|\}.
\ee
With these quantities we define the sets
\be
B_i := \{r_i \leq R_i\},
\ee
\be
B := \{r \geq R\}.  
\ee
It is straightforward to show that
\be
B_i \cap B = \emptyset, \qquad B_i \cap B_j = \emptyset,\qquad i\neq j.
\ee
Therefore each set $B_i$, $B$ contains only one end, and we show that the norms restricted to each of these sets are equivalent, which is a direct corollary to the following proposition.
\begin{proposition}\label{prop1}
If $r_i\leq R_i$, then there is a constant $C_i$, $0<C_i<1$, such that
 \be\label{ineqw1}
 C_i r_i \leq w < r_i.
\ee
If $r \geq R$, then
\be\label{ineqw2}
\frac{1}{2N}r < w < \frac{2}{N}r.
\ee
\end{proposition}
\begin{proof}
We first prove \eqref{ineqw1}. The right hand side inequality follows from 
the definition of $w$
\be\label{desIzq}
\frac{1}{w} = \sum_j \frac{1}{r_j} > \frac{1}{r_i},\qquad \forall\, i=1,\ldots,N.
\ee 
For the left hand side inequality, as $r_i \leq R_i$, then for $j\neq i$
\be
r_j \geq d_{ij} - R_i > 0
\ee
and
\begin{eqnarray}
\frac{1}{w} & = & \frac{1}{r_i} + \sum_{j\neq i} \frac{1}{r_j} \leq  \frac{1}{r_i} + \sum_{j\neq i} \frac{1}{d_{ij} - R_i} \\
& \leq & \frac{1}{r_i} + \left(\sum_{j\neq i} \frac{1}{d_{ij} - R_i}\right)\frac{R_i}{r_i} 
 =  \frac{1}{r_i}\left(1 + R_i \sum_{j\neq i} \frac{1}{d_{ij} - R_i}\right).
\end{eqnarray}
If we define
\be
C_i := \left(1 + R_i \sum_{j\neq i} \frac{1}{d_{ij} - R_i}\right)^{-1}
\ee
the proof of \eqref{ineqw1} is complete.

Now we prove \eqref{ineqw2}. If $r\geq R$,
\be
\frac{r}{2}\leq r-\frac{R}{2} < r_i < r+\frac{R}{2} <2r,
\ee
therefore
\be
\frac{1}{2r} < \frac{1}{r_i} < \frac{2}{r},
\ee
and
\be
\frac{N}{2r} < \frac{1}{w} < \frac{2N}{r},
\ee
which gives \eqref{ineqw2}.
\end{proof}
For the corollary we denote by $||\cdot||'_{w,k,p,\delta;B_i}$ the norm with the domain of integration restricted to $B_i$ (resp. for the domain $B$). In the case of the standard norms, $||\cdot||'_{k,p,\delta;B_i}$ also means that in the integration the replacement $r\ra r_i$ has been made, centering the norm at $x_i$. 
\begin{corollary}\label{equivNorms}
 The norms $||\cdot||'_{k,p,\delta;B_i}$ and $||\cdot||'_{w,k,p,\delta;B_i}$ (resp. $||\cdot||'_{k,p,\delta;B}$ and $||\cdot||'_{w,k,p,\delta;B}$) are equivalent.
\end{corollary}

Next we prove important estimates on functions $u\in W'^ {k,p}_{w,\delta}$ near the asymptotic ends. This is analogous to and follows Lemma A.1 in \cite{Dain:2010uh}.
\begin{lemma}\label{decay}
Assume $u\in W'^{k,p}_{w,\delta}$ with  $n-kp<0$, then we have the following estimate
 \begin{equation}
   \label{eq:41}
   w^{-\delta}|u|\leq C \left\lVert u  \right\rVert'_{w,k,p,\delta}. 
 \end{equation}
Moreover, we have
 \begin{equation}
\label{eq:44}
 \lim_{r_i\to 0}w^{-\delta}  |u|= \lim_{r\to \infty}w^{-\delta}  |u|= 0. 
\end{equation}
\end{lemma}

\begin{proof}
Restricting the domain of $u$ to $B_i$ there are constants $C_k$ such that
\be
w^{-\delta}|u| \leq C_1 r_i^{-\delta}|u|\leq C_2 ||u||'_{k,p,\delta;B_i} \leq C_3 ||u||'_{w,k,p,\delta;B_i} \leq C_4 ||u||'_{w,k,p,\delta},
\ee
where we have used first Proposition \ref{prop1}, then Lemma A.1 in \cite{Dain:2010uh}, noticing that in that lemma the inequality and the 
limits can also be proven separately for a ball around the origin and for the asymptotic region $r>R$, therefore we can use separately the
bounds in each of $B_i$ and in $B$. For the third inequality Corollary \ref{equivNorms} was used and the last inequality comes 
from the fact that the norm of $u$ on the whole domain is bigger than its norm restricted to $B_i$. The same sequence of inequalities 
hold for the domain $B$. Also, the set
\be
A = \{r_i \geq R_i, r\leq R\},
\ee
is compact and hence the function
\be
w^{-\delta}|u|
\ee
has a maximum there, this together with the previous inequalities means that there exists a constant $C$ such that \eqref{eq:41} 
is satisfied in $\mR^n\backslash \Pi$. 
Also, using the same argument as in Lemma A.1 in \cite{Dain:2010uh}, we have \eqref{eq:44}.
\end{proof}

Another important result that will be useful later is the following.

\begin{lemma}\label{lemab1}
 If $u\in W'^{1,p}_{w,-1/2}$, then $w^{\frac{1}{2}-\frac{n}{p}}u\in L^p$ and $w^{\frac{3}{2}-\frac{n}{p}}\partial u \in L^p$.
\end{lemma}
\begin{proof}
\begin{equation}
 \|u\|_{W'^{1,p}_{w,-1/2}} = \|u\|_{L'^p_{w,-1/2}} + \|\partial u\|_{L'^p_{w,-3/2}} = \|w^{\frac{1}{2}-\frac{n}{p}}u\|_{L^p} + \|w^{\frac{3}{2}-\frac{n}{p}}\partial u\|_{L^p},
 \end{equation}
 and therefore if $\|u\|_{W'^{1,p}_{w,-1/2}}$ is bounded, then $\|w^{\frac{1}{2}-\frac{n}{p}}u\|_{L^p}$ and $\|w^{\frac{3}{2}-\frac{n}{p}}\partial u\|_{L^p}$ are also bounded.
\end{proof}

The main result we prove in this section is the fact that the Laplace operator is an isomorphism $\Delta: W'^{k+2,p}_{w,\delta}\ra W'^{k,p}_{w,\delta-2}$. We  follow 
Bartnik's arguments in the proof of Theorem 1.7 in \cite{Bartnik86}. We need to exclude some exceptional values of $\delta$. The exceptional values are $\{ m'\in\mathbb{Z}, m'\neq -1,-2,\ldots,3-n\}$, $\delta$ is said to be nonexceptional if it is not one of those values, and it is convenient to define
\be
 m = \max \,\{ m' \mbox{ exceptional}, m'<\delta \}.
\ee

\begin{theorem}\label{lemaIsom}
 Let $\delta$ be nonexceptional, $1<p<\infty$, and $k$ a non-negative integer. Then the Laplace operator
 \be
 \Delta: W'^{k+2,p}_{w,\delta}\ra W'^{k,p}_{w,\delta-2}
 \ee
 is an isomorphism and there exists a constant $C=C(n,p,\delta,k)$ such that
 \be \label{bound1}
 \|u\|'_{w,k+2,p,\delta}\leq C\|\Delta u\|'_{w,k,p,\delta-2}.
 \ee
\end{theorem}
\begin{proof}
As noted in \cite{Bartnik86} we only need to prove the $k=0$ case. We start considering the convolution kernel $K(x,y)$,
 \be\label{kernelLap}
 K(x,y) = \left\{ \begin{array}{ll}
                   |x-y|^{2-n}, & \mbox{if }2-n<\delta<0,\\
                   |x-y|^{2-n}-|y|^{2-n}\sum_{j=0}^m P_j^\lambda\left(\frac{x\bullet y}{|x||y|}\right)\left(\frac{|x|}{|y|}\right)^j, & \mbox{if } m\geq 0,\\
                   |x-y|^{2-n}-|x|^{2-n}\sum_{j=0}^m P_j^\lambda\left(\frac{x\bullet y}{|x||y|}\right)\left(\frac{|y|}{|x|}\right)^j, & \mbox{if } m<2-n,
                  \end{array}\right.
 \ee
 where $x\bullet y$ is the Euclidean scalar product of $x$ and $y$, and $P_j^\lambda$ are the ultraspherical functions. Associated with this kernel we have the operator $K$,
 \be
 Ku(x) = \int_{\mR^n\backslash\Pi} K(x,y)u(y)\,dy.
 \ee
 We present only in detail the first case in \eqref{kernelLap}, as the three cases are similar and our main interest concerns $n=3$, $\delta=-1/2$.
Let us define the associated kernel
 \be
 K'(x,y) = w(x)^{-\delta-n/p}\, K(x,y)\, w(y)^{\delta-2+n/p}.
 \ee
Then, the Nirenberg-Walker Lemma 2.1 in \cite{Nirenberg73} implies that $K'$ is a bounded $L^p\to L^p$ operator, and the equality
 \be
 \|Ku\|'_{w,p,\delta} = \|K'(w^{-\delta+2-n/p}u)\|_{L^p},
 \ee
means that
 \be
 \|Ku\|'_{w,p,\delta} = \|K'(w^{-\delta+2-n/p}u)\|_{L^p} \leq C\|w^{-\delta+2-n/p}u\|_{L_p} = C \|u\|'_{w,p,\delta-2},
 \ee
 where the last equality comes directly from the definition of the norms. This shows that $K$ is a bounded $L'^p_{w,\delta-2}\rightarrow L'^p_{w,\delta}$ operator.
The distributional identities
\be
\Delta_xK(x,y)=\Delta_yK(x,y)=\delta(x-y) \qquad \mbox{ in }\mR^n\backslash\Pi
\ee
imply that $K(\Delta u)=u$ for all $u\in C_c^\infty(\mR^n\backslash\Pi)$,
and therefore
\be
\|u\|_{L'^p_{w,\delta}}=\|K(\Delta u)\|_{L'^p_{w,\delta}}\leq C \|\Delta u\|_{L'^p_{w,\delta-2}}\qquad \forall u\in W'^{2,p}_{w,\delta}.
\ee
This and the elliptic estimate 
\be\label{elli}
\|u\|_{W'^{2,p}_{w,\delta}}\leq c(\|\Delta u\|_{L'^p_{w,\delta-2}}+\|u\|_{L'^p_{w,\delta}})
\ee
give the desired estimates \eqref{bound1}. Note that the proof of \eqref{elli} uses the same arguments as in the case of standard, not weighted Sobolev spaces.
The rest of the proof follows the same lines as Bartnik's.
\end{proof}

\section{Axially symmetric multi-black hole initial data}\label{initial}

The main application of the Sobolev spaces presented in section \ref{Sobolev} that we want to discuss here is on the deformation of extremal black hole initial data.  
We refer the reader to \cite{Acena:2015rxa} where the one black hole case is treated. The general argument and several calculations can readily be extended from one to 
several black holes, and we present here mainly the path needed to understand the theorem and its proof and the arguments that need special care, some particular details are left to the appendix.

Axial symmetry is not required for the validity of the main arguments presented below, but it simplifies the treatment of certain bounds and asymptotic behaviours. Also it is 
necessary when dealing with quasilocal angular momentum. The case without axial symmetry will be given elsewhere.

From now on we assume axial symmetry, namely, there is a spacelike Killing vector field, $\eta$, on $M$, with complete closed orbits, such that
\begin{equation}\label{axSym}
\c{L}_\eta g_{ij}=0,\quad\c{L}_\eta K_{ij}=0,\quad\c{L}_\eta E^i=0,\quad\c{L}_\eta B^i=0,
\end{equation}
and as we want the axial symmetry to be reflected in the conformal metric,
\begin{equation}
 \c{L}_\eta\Phi = 0.
\end{equation}
Physically this condition implies that the black holes are all aligned, situated on the symmetry axis $\Gamma$. The perturbed initial data will also have this property.

We also take the initial data to be time-rotation symmetric, which implies that the data is maximal
\begin{equation}
 K = 0 \Rightarrow \t{K}=0,
\end{equation}
and that $(M,g_{ij})$ is in the positive Yamabe class. We use cylindrical coordinates $(z,\rho,\phi)$ adapted to the axial symmetry
\be
\eta = \partial_\phi.
\ee
Using the symmetry it can be shown that there are scalar potentials $\omega$, $\psi$, $\chi$ that contain the information about the extrinsic curvature and the electromagnetic fields, 
and that they do not depend on $\phi$. Also, all the constraint equations but \eqref{constEq} are automatically satisfied.

We define the conformal metric as
\begin{equation}
 \t{g}_{ij} = e^{2q}(dz^2+d\rho^2)+\rho^2 d\phi^2,
\end{equation}
where $q=q(z,\rho)$. Then,
\eqref{constEq} takes the form
\begin{eqnarray}\label{LicEq}
 \Delta\Phi = -\frac{1}{4}\Delta_2q\Phi - \frac{(\partial\omega)^2}{16\rho^4\Phi^7}-\frac{(\partial\psi)^2+(\partial\chi)^2}{4\rho^2\Phi^3},
\end{eqnarray}
where
\begin{equation}
 \Delta_2 := \partial_z^2 + \partial_\rho^2,\quad
 \Delta := \partial_z^2 + \partial_\rho^2 + \rho^{-1}\partial_\rho,
\end{equation}
and consider that $\Phi = \Phi(z,\rho)$.

Given a set of functions $(\Phi_0,q_0,\omega_0,\psi_0,\chi_0)$ that satisfy \eqref{LicEq}, we take arbitrary smooth axially symmetric functions with compact support outside the
symmetry axis, $q,\,\omega,\,\psi,\,\chi$, and a parameter $\lambda$, perturb the given functions as follows
\begin{equation}
 q \ra q_0+\lambda q, \quad \omega \ra \omega_0+\lambda \omega, \quad \psi \ra \psi_0+\lambda \psi, \quad \chi \ra \chi_0+\lambda \chi,
\end{equation}
and write the perturbed solution as
\begin{equation}
 \Phi = \Phi_0 + u.
\end{equation}
Then equation \eqref{LicEq} can be written as an equation for $u$, namely
\begin{eqnarray}\label{LicEqPert}
 \Delta u & = & -\frac{1}{4}\Delta_2q_0 u -\frac{1}{4}\lambda \Delta_2 q(\Phi_0+u) - \frac{(\partial\omega_0+\lambda\partial\omega)^2}{16\rho^4(\Phi_0+u)^7}+\frac{(\partial\omega_0)^2}{16\rho^4\Phi_0^7} \\ 
&& - \frac{(\partial\psi_0+\lambda\partial\psi)^2}{4\rho^2(\Phi_0+u)^3}+\frac{(\partial\psi_0)^2}{4\rho^2\Phi_0^3} - \frac{(\partial\chi_0+\lambda\partial\chi)^2}{4\rho^2(\Phi_0+u)^3}+\frac{(\partial\chi_0)^2}{4\rho^2\Phi_0^3}.
\end{eqnarray}
In order to consider existence of solutions we associate with this equation the map $G$,
\begin{eqnarray}
 G(\lambda,u) & := & \Delta u + \frac{1}{4}\Delta_2q_0 u + \frac{1}{4}\lambda \Delta_2 q(\Phi_0+u) + \frac{(\partial\omega_0+\lambda\partial\omega)^2}{16\rho^4(\Phi_0+u)^7} - \frac{(\partial\omega_0)^2}{16\rho^4\Phi_0^7} \\ 
&& + \frac{(\partial\psi_0+\lambda\partial\psi)^2}{4\rho^2(\Phi_0+u)^3} - \frac{(\partial\psi_0)^2}{4\rho^2\Phi_0^3} + \frac{(\partial\chi_0+\lambda\partial\chi)^2}{4\rho^2(\Phi_0+u)^3} - \frac{(\partial\chi_0)^2}{4\rho^2\Phi_0^3},
\end{eqnarray}
where solutions of \eqref{LicEqPert} are given by
\begin{equation}\label{eqG}
 G(\lambda,u) = 0.
\end{equation}
In particular, the background solution is
\begin{equation}
 G(0,0)=0.
\end{equation}

The main result of this section is the following theorem, an extension of theorem 2.1 of \cite{Acena:2015rxa}. We use
the standard notation $H'^k_{w,\delta} = W'^{k,2}_{w,\delta}$.

\begin{theorem}\label{thm}
Let $q, \omega, \psi, \chi\in C^\infty_0(\mathbb R^3\setminus\Gamma)$ be arbitrary smooth axially symmetric functions. Then,
there is $\lambda_0>0$ such that for all $\lambda\in (-\lambda_0, \lambda_0)$  there
exists a solution $u(\lambda)\in H'^{2}_{w,-1/2}$ of equation
\eqref{eqG}. The solution $u(\lambda)$ is continuously differentiable in
$\lambda$ and satisfies $\Phi_0+u(\lambda)>0$. Moreover, for small $\lambda$
and small $u$ (in the norm $ H^{'2}_{w,-1/2}$) the solution $u(\lambda)$ is the
unique solution of equation \eqref{eqG}.   
\end{theorem}

\begin{proof}
The proof of this theorem follows the lines of \cite{Acena:2015rxa} and makes strong use of the Inverse Function theorem. We will ommit some arguments that are carried over from one cylindrical end to many such ends. 
The proof can be thought of as showing that $u$ gives rise to a good conformal factor and collecting the necessary conditions for the Inverse Function theorem to hold.

\begin{itemize} 
\item \textit{$\Phi_0+u$ is positive.}
  
We are considering the map $G:\mathbb{R}\times H'^2_{w,-1/2}\ra L'^2_{w,-5/2}$, but for a general $u\in H'^2_{w,-1/2}$ 
the function $\Phi=\Phi_0+u$ does not have a definite sign. In order for $\Phi$ to be a 
conformal factor we need it to be positive. As we assume $\Phi_0$ to be a  conformal factor, then we can 
conjecture that if $u$ is small enough, then $\Phi$ is also going to be a  conformal factor. This can be achieved by restricting $u$ to be in a ball around the origin in $H'^2_{w,-1/2}$. We therefore define the subset $V$ of $H'^2_{w,-1/2}$ as
\begin{equation}
 V = \{v\in H'^2_{w,-1/2} : ||v||_{H'^2_{w,-1/2}}<\xi\},
\end{equation}
where $\xi>0$ is a constant that can always be chosen and whose particular value depends on $\Phi_0$. We restrict the map, 
$G:\mathbb{R}\times V\ra L'^2_{w,-5/2}$, having that if $u\in V$ then
\begin{equation}
 \Phi_0+u > 0.
\end{equation}

\item\textit{The map $G$ is well defined}

To prove that the map $G:\mathbb{R}\times V\ra L'^2_{w,-5/2}$ is well defined we need to show that $||G(\lambda,u)||_{L'^2_{w,-5/2}}$ is bounded whenever $\lambda \in \mathbb{R}$ and $u\in V$. This is accomplished using the definition of the $H'^2_{w,-1/2}$ norm, the triangle inequality, the asymptotic conditions on the background functions and the compact support of $q$, $\omega$, $\psi$ and $\chi$,  together with the inequalities presented in Appendix \ref{useIneq}.

\item \textit{Partial Fr\'echet derivatives of $G$}

To obtain the partial Fr\'echet derivatives of $G$ and to show that $G$ is $C^1$ we start by calculating the directional derivatives of $G$, namely
\begin{equation}
 \left.\frac{d}{dt}G(\lambda+t\gamma,u)\right|_{t=0},\qquad \frac{d}{dt}G(\lambda,u+tv)\Bigg|_{t=0},
\end{equation}
and propose them as the partial Fr\'echet derivatives of $G$,
\begin{eqnarray}
 D_1G(\lambda,u)[\gamma] & = &\Bigg[\frac{1}{4}\Delta_2q(\Phi_0+u)+\frac{\partial\omega(\partial\omega_0+\lambda\partial\omega)}{8\rho^4(\Phi_0+u)^7} \\
 && +\frac{\partial\psi(\partial\psi_0+\lambda\partial\psi)}{2\rho^2(\Phi_0+u)^3}+\frac{\partial\chi(\partial\chi_0+\lambda\partial\chi)}{2\rho^2(\Phi_0+u)^3}\Bigg]\gamma,\\
 D_2G(\lambda,u)[v] & = & \Delta v + \Bigg[\frac{1}{4}\Delta_2 q_0+\frac{1}{4}\lambda\Delta_2q - \frac{7(\partial\omega_0+\lambda\partial\omega)^2}{16\rho^4(\Phi_0+u)^8} \\
 && -\frac{3(\partial\psi_0+\lambda\partial\psi)^2}{4\rho^2(\Phi_0+u)^4}-\frac{3(\partial\chi_0+\lambda\partial\chi)^2}{4\rho^2(\Phi_0+u)^4}\Bigg]v.
\end{eqnarray}
Now it can be shown that these operators are bounded using the properties mentioned in the point above, showing that they are linear operators between the following spaces
\begin{equation}
D_1G(\lambda,u):\mR\ra L'^2_{-5/2}, \qquad D_2G(\lambda,u):H'^2_{-1/2}\ra L'^2_{-5/2}.
\end{equation}
A lengthy but straightforward calculation shows that
\begin{equation}
 \lim_{\gamma\ra0}\frac{||G(\lambda+\gamma,u)-G(\lambda,u)-D_1G(\lambda,u)[\gamma]||_{L'^2_{-5/2}}}{|\gamma|}=0,
\end{equation}
\begin{equation}
 \lim_{v\ra0}\frac{||G(\lambda,u+v)-G(\lambda,u)-D_2G(\lambda,u)[v]||_{L'^2_{-5/2}}}{||v||_{H'^2_{-1/2}}}=0
\end{equation}
and therefore $D_1G$ and $D_2G$ are indeed the partial Fr\'echet derivatives of $G$. It can also be checked that 
the derivatives are continuous, and then $G$ is $C^1$.

\item \textit{The map $D_2G(0,0):H'^2_{w,-1/2}\ra L'^2_{w,-1/2}$ is an isomorphism}

Associated with the map $D_2G(0,0)$ we define the operator $\mathcal{L}$ through
\be
\mathcal{L}v := -\Delta v+\alpha v,
\ee
where
\begin{equation}\label{alpha0}
 \alpha=-\frac{\Delta_2q_0}{4}+7\frac{(\partial\omega_0)^2}{16\rho^4\Phi_0^8}+3\frac{(\partial\psi_0)^2+(\partial\chi_0)^2}{4\rho^2\Phi_0^4}.
\end{equation}
We want to show that the equation
\begin{equation}\label{ecL}
\mathcal Lu:=-\Delta u+\alpha u=f \qquad \mbox{in}\; \mathbb R^{3}\setminus \Pi
\end{equation}
has a unique solution $u\in H'^2_{w,-1/2}$ for each $f\in L'^2_{w,-1/2}$. The steps for this are the same as in \cite{Acena:2015rxa},
see also the references therehin. The steps of the proof are as follows. First the Yamabe condition is used to show that for all $f\in C_0^\infty$, $f\neq0$,
\be\label{yamflat} 
 \int_M|\partial f|^2+\alpha f^2 d\mu>0,
 \ee
where $\alpha$ is given by \eqref{alpha0} and  the norm and volume
element in \eqref{yamflat} are computed with respect to the flat metric. This is used to show that the bilinear form
\begin{equation}\label{Bil}
B[u,v]:=\int_{\Rt\setminus\Pi}\partial u\cdot \partial v+\alpha uv\,d\mu,
\end{equation}
with $u,v\in H^{'1}_{w,-1/2}$, which is well defined in virtue of Lemma \ref{lemab1} and corresponds to 
the linear operator $\mathcal L$, is indeed an inner product. Using equation \eqref{alpha0} and  the standard Holder inequality for $L^2$ spaces (note that by Lemma \ref{lemab1}, $\partial u,\;\partial v,\; uw^{-1},\;v^{-1}\in L^2$) it is shown that the linear functional $\ell(\cdot):=B[\cdot,v]$ is 
bounded for all $v\in H^{'1}_{w,-1/2}$. With these conditions fulfilled, the Riesz Representation Theorem states that there exists a unique weak solution, $u\in H^{'1}_{w,-1/2}$, of $\mathcal Lu=f$, for each $f\in L^{'2}_{w,-5/2}$. The last step is to show that the weak solution is indeed in $H^ {'2}_{w,-1/2}$. This follows the same argument as in \cite{gabach09}, the only new ingredient is the use of Theorem \ref{lemaIsom}.
\end{itemize}
This completes the proof, as we have shown that the conditions for the Inverse Function theorem to hold are true in the case at hand.
\end{proof}

\appendix

\section{Bounds and inequalities}\label{useIneq}

In this appendix we collect some inequalities which have been implicitly used in the text and that we consider noteworthy enough to keep them 
for reference. The explicit calculations are very similar to those in \cite{Dain:2010uh,gabach09,Acena:2015rxa}, and therefore we do not include them.

Due to $\Phi_0$ being the conformal factor and encoding the metric behaviour at the asymptotically flat end and at the cylindrical ends, we have that there are positive constants $C_1$, $C_2$, $C_3$ and $C_4$ such that
\begin{equation}
 C_1\sqrt{w+C_2}\leq\sqrt{w}\Phi_0\leq C_3\sqrt{w+C_4}.
\end{equation}
Also, for $u,v\in V$ there are bounded functions $H_1$, $H_2$, $H_3$, such that
\begin{equation}
 \frac{1}{\Phi_0^p}-\frac{1}{(\Phi_0+u)^p} = u\,w^{\frac{p+1}{2}} H_1,
\end{equation}
\begin{equation}
 \frac{1}{(\Phi_0+u)^p}-\frac{1}{(\Phi_0+v)^p} = (v-u)\,w^{\frac{p+1}{2}} H_2,
\end{equation}
\begin{equation}
 \frac{1}{(\Phi_0+u+v)^p}-\frac{1}{(\Phi_0+u)^p}+\frac{pv}{(\Phi_0+u)^{p+1}} = w^{\frac{p+2}{2}}v^2 H_3,
\end{equation}
which means that there are positive constantes $C_5$, $C_6$, $C_7$, such that
\be\label{desH}
\left|\frac{1}{\Phi_0^p}-\frac{1}{(\Phi_0+u)^p}\right| \leq C_5\,|u|\,w^{\frac{p+1}{2}}.
\ee
\be
\left|\frac{1}{(\Phi_0+u)^p}-\frac{1}{(\Phi_0+v)^p}\right| \leq C_6\, |v-u|\,w^{\frac{p+1}{2}}.
\ee
\begin{equation}
 \left|\frac{1}{(\Phi_0+u+v)^p}-\frac{1}{(\Phi_0+u)^p}+\frac{pv}{(\Phi_0+u)^{p+1}}\right| \leq C_7 w^{\frac{p+2}{2}} |v|^2.
\end{equation}

\bibliographystyle{plain}

\bibliography{biblio}

\begin{thebibliography}{10}

\bibitem{Acena:2015rxa}
Andr\'es Ace\~na and Mar\'ia E.~Gabach Cl\'ement.
\newblock {Extremal black hole initial data deformations}.
\newblock {\em Class. Quant. Grav.}, 33(11):115017, 2016.

\bibitem{Bartnik86}
R.~Bartnik.
\newblock The mass of an asymptotically flat manifold.
\newblock {\em Comm. Pure App. Math.}, 39:661--693, 1986.

\bibitem{Choquet99}
Yvonne Choquet-Bruhat, James Isenberg, and James~W. York, Jr.
\newblock {E}instein constraint on asymptotically euclidean manifolds.
\newblock {\em Phys. Rev. D}, 61:084034, 1999.

\bibitem{Chrusciel:2010my}
Piotr~T. Chrusciel, Justin Corvino, and James Isenberg.
\newblock {Construction of N-body initial data sets in general relativity}.
\newblock {\em Commun. Math. Phys.}, 304:637--647, 2011.

\bibitem{Chrusciel02a}
Piotr~T Chru{\'s}ciel and Rafe Mazzeo.
\newblock On `many-black-hole' vacuum spacetimes.
\newblock {\em Class. Quantum Grav.}, 20:729--754, 2003.

\bibitem{Chrusciel:2012np}
Piotr~T. Chru{\'{s}}ciel and Rafe Mazzeo.
\newblock {Initial data sets with ends of cylindrical type: I. The Lichnerowicz
  equation}.
\newblock {\em Ann. Henri Poincar\'e}, 16:1231--1266, 2015.

\bibitem{dain10d}
Sergio Dain.
\newblock Extreme throat initial data set and horizon area-angular momentum
  inequality for axisymmetric black holes.
\newblock {\em Phys. Rev. D}, 82:104010, 2010.

\bibitem{Dain:2010uh}
Sergio Dain and Mar\'ia~Eugenia Gabach~Cl\'ement.
\newblock {Small deformations of extreme {Kerr} black hole initial data}.
\newblock {\em Class. Quantum Grav.}, 28:075003, 2011.

\bibitem{gabach09}
Mar\'ia~Eugenia Gabach~Cl\'ement.
\newblock Conformally flat black hole initial data with one cylindrical end.
\newblock {\em Class. Quantum Grav.}, 27:125010, 2010.

\bibitem{Leach:2016jra}
Jeremy Leach.
\newblock {Non-constant mean curvature trumpet solutions for the Einstein
  constraint equations}.
\newblock {\em Class. Quant. Grav.}, 33(14):145001, 2016.

\bibitem{Majumdar47}
S.D. Majumdar.
\newblock {A class of exact solutions of Einstein's field equations}.
\newblock {\em Phys. Rev.}, 72:390, 1947.

\bibitem{Maxwell:2004yb}
David Maxwell.
\newblock {Rough solutions of the Einstein constraint equations}.
\newblock 2004.

\bibitem{Maxwell:2005bd}
David Maxwell.
\newblock {Rough solutions of the Einstein constraint equations on compact
  manifolds}.
\newblock {\em J. Hyperbol. Diff. Equat.}, 2:521, 2005.

\bibitem{Nirenberg73}
L.~Nirenberg and H.~Walker.
\newblock The null spaces of elliptic partial differential operators in
  $\mathbb{R}^n$.
\newblock {\em J. Math. Anal. and Appl}, 42:271--301, 1973.

\bibitem{Papapetrou45}
A.~Papapetrou.
\newblock {A static solution of the equations of the gravitational field for an
  arbitrary charge-distribution}.
\newblock {\em Proc. R. Irish Acad. A}, 51:191--204, 1945.

\end{thebibliography}

\end{document}